\documentclass[pdftex,11pt]{article}
\usepackage[T1]{fontenc}
\usepackage{ae,aecompl}
\usepackage{amssymb,amsmath}

\usepackage{fullpage}
\usepackage{color}
\usepackage{authblk}
\newcommand{\greg}[1]{#1}

\def\idrm#1{\ensuremath{\mathrm{#1}}}
\def\idtt#1{\ensuremath{\mathtt{#1}}}

\def\ceil#1{\lceil #1 \rceil}

\newcommand{\no}[1]{}

\newtheorem{theorem}{Theorem}
\newtheorem{lemma}{Lemma}
\newenvironment{proof}{\trivlist\item[]\emph{Proof}:}%
{\unskip\nobreak\hskip 1em plus 1fil\nobreak$\Box$
\parfillskip=0pt%
\endtrivlist}

\newenvironment{itemize*}%
  {\begin{itemize}%
    \setlength{\itemsep}{0pt}%
    \setlength{\parskip}{0pt}%
    \setlength{\parsep}{0pt}%
    \setlength{\topsep}{0pt}%
    \setlength{\partopsep}{0pt}%
  }%
  {\end{itemize}}%

\newcommand{\cL}{{\cal L}}
\newcommand{\cT}{{\cal T}}
\newcommand{\cP}{{\cal P}}

\newcommand{\bT}{\mathbf{T}}

\newcommand{\Trev}{\overleftarrow{T}}

\newcommand{\new}{\idrm{new}}
\newcommand{\col}{\idtt{col}}

\begin{document}
\title{Full-fledged Real-Time Indexing for Constant Size
  Alphabets\thanks{This version is an update of the paper published in
  ICALP'13 proceedings. We strengthen the main result (Theorem~\ref{maintheorem}) by replacing the
  {\em expected} worst-case time by deterministic worst-case
  time. This is achieved by using a simpler suffix tree update
  construction in Section~\ref{sec:onlinenew} (Theorem~\ref{update-time}).}}

\author[1,2]{Gregory Kucherov}
\author[3]{Yakov Nekrich\thanks{This work was done while this author was at 
Laboratoire d'Informatique Gaspard Monge, Universit\'e Paris-Est \& CNRS}}
\affil[1]{Laboratoire d'Informatique Gaspard Monge, Universit\'e Paris-Est \& CNRS, Marne-la-Vall\'ee, Paris, France, \texttt{Gregory.Kucherov@univ-mlv.fr}}
\affil[2]{Department of Computer Science, Ben-Gurion University of the Negev, Be'er Sheva, Israel}
\affil[3]{Department of Electrical Engineering \& Computer Science,
  University of Kansas, \texttt{yakov.nekrich@googlemail.com}}

\date{\empty}

\maketitle

\begin{abstract}
In this paper we describe a data structure that supports pattern
matching queries on a dynamically arriving text over an alphabet of
constant size. Each new symbol can be prepended to $T$ in $O(1)$
worst-case
time. At any moment, we can report all 
occurrences of a pattern $P$ in the current text in $O(|P|+k)$ time, where $|P|$ is the
length of $P$ and $k$ is the number of occurrences.  This
resolves, under assumption of constant size alphabet,
a long-standing open problem of existence of a real-time
indexing method for string matching (see \cite{AmirN08}). 
\end{abstract}

\section{Introduction}
Two main
versions of the string matching problem differ in which of the two components -- 
pattern $P$ or text $T$ -- is provided first in the input (or is
considered as fixed) and can then be
preprocessed before processing the other component. The framework when
the text has to be preprocessed is usually called {\em indexing}, as
it can be viewed as constructing a text index supporting matching
queries. 

Real-time variants of the string matching problem are about as old as
the string matching itself. In the 70s, existence of real-time string
matching algorithms was first studied for Turing machines. For
example, it has been shown that the language $\{P\#T|P\mbox{ occurs in }T\}$ can be recognized by a Turing machine, while the
language $\{T\#P|P\mbox{ occurs in }T\}$ cannot \cite{DBLP:journals/jacm/Galil81}. In the realm of the RAM
model, the real-time variant of pattern-preprocessing string matching
has been extensively studied, leading to very efficient solutions
(see e.g. \cite{BreslauerGrossiMignosiCPM11} and references therein). The
indexing variant, however, still has important unsolved questions. 

Back in the 70s, Slisenko \cite{DBLP:conf/mfcs/Slisenko78} claimed a
real-time algorithm for recognizing the language $\{T\#P|P\mbox{
  occurs in }T\}$ on the RAM model, but its complex and voluminous full description
made it unacknowledged by the scientific community, and 
the problem remained to be considered open
for many years. In 1994, Kosaraju \cite{Kosaraju94}
reported another solution to this problem. In our present work, however, we
are interested in a more general problem, when matching queries can be
made at all moments, rather than after the entire text has been
received. Specifically, in our problem, a streaming text 
should be processed in real time so that at
each moment, a matching query $P$ can be made to the portion of the
text received so far. We call this the {\em real-time indexing
  problem}. 
This problem has been considered in 2008 by Amir and Nor
\cite{AmirN08}, who extended Kosaraju's algorithm to deal with
repetitive queries made at any moment of the text scan. 

{All the three existing real-time indexing solutions
  \cite{DBLP:conf/mfcs/Slisenko78,Kosaraju94,AmirN08} support only existential queries
  asking whether the pattern occurs in the text, but are unable to
  report occurrences of the pattern.} Designing a real-time text indexing
algorithm that would
support queries on all occurrences of a pattern is stated in
\cite{AmirN08} as the {most important remaining open problem}. 
The algorithms of
\cite{Kosaraju94,AmirN08}  assume a constant size alphabet and are both based on constructions of
``incomplete'' suffix trees which can be built real-time but can only
answer existential queries.  To output all occurrences of a pattern, a
fully-featured suffix tree is needed, however a real-time suffix tree
construction, first studied in \cite{AmirKopelLewen2SPIRE05}, 
is in itself an open question. The best
currently known algorithms spend on each character $O(\log\log n)$
worst-case time
in the case of constant-size alphabets \cite{BreslauerI11}, or $O(\log\log
n+\frac{\log^2\log |A|}{\log\log\log |A|})$ time for 
arbitrary alphabets $A$ \cite{FischerG13}.  A truly real-time suffix
tree construction seems unlikely to
exist. Therefore, a suffix tree alone seems to be insufficient to
solve the real-time indexing problem. 

{In this paper, we propose the first real-time text indexing  solution
that supports reporting all pattern occurrences, under the assumption
of constant size alphabet.} 
Similar to the
previous works on real-time indexing, we
assume that the text is read right-to-left, or otherwise the 
pattern needs to be
reversed before executing the query. 
The general idea is
to maintain several data structures, three in our case, each
supporting queries for different pattern lengths. Each of these
structures is based on a suffix tree (or suffix-tree-like structure)
exteded by some auxiliary data structures. 
To update a suffix tree, we use an implementation of Weiner's
algorithm which is somewhat similar to but simpler than that of
\cite{BreslauerI11}. The simplification is achieved by using some
external algorithmic tools, such as colored predecessor queries
\cite{Mortensen03}. As a result, we can update a suffix tree in $O(\log\log n)$
worst-case time per letter, under the assumption that alphabet size is
bounded by $O(\log^{1/4}n)$ and without resorting to a
deamortization as in \cite{BreslauerI11}. This is an interesting result in itself. 

The paper is organized as follows. 
In Section~\ref{preliminaries},
we describe auxiliary data
structures and present our method for online update of suffix trees. 
In Section~\ref{sec:algo}, we describe the three data structures for
different pattern lengths that constitute a basis of our solution. These
data structures, however, do not provide a fully real-time algorithm. 
Then in Section~\ref{sec:realtime}, we show how to ``fix'' the
solution of Section~\ref{sec:algo} in order to obtain a fully
real-time algorithm. 

Throughout the paper, $\Sigma$ is an alphabet of constant size
$\sigma$. 
Since the text $T$ is read right-to-left, it will be convenient
for us to enumerate symbols of $T$ from the end, i.e. $T=t_n\ldots t_1$ and
substring $t_{i+\ell}t_{i+\ell-1}\ldots t_i$ will be denoted
$T[i+\ell..i]$. $T[i..]$ denotes suffix $T[i..1]$. 
Throughout this paper, we reserve $k$ to denote the number of objects
(occurrences of a pattern, elements in a list, etc) in the query
answer. 

\section{Preliminaries}
\label{preliminaries}
In this Section, we describe main algorithmic tools used by our
algorithms. 
\subsection{Range Reporting and Predecessor Queries on Colored Lists}
\label{sec:prelim}
We use 
 data structures
from~\cite{Mortensen03} for 
searching in dynamic colored lists. 

\paragraph{Colored Range Reporting in a List.}
Let elements of a dynamic linked list $\cL$ be assigned positive integer values called \emph{colors}. A colored range reporting query on a list $\cL$ consists of two integers $col_1<col_2$ and two pointers $ptr_1$ and $ptr_2$ that point to elements $e_1$ and $e_2$  of $\cL$. 
An answer to a colored range reporting query consists of all elements
$e\in \cL$ occurring between $e_1$ and $e_2$ (including $e_1$ and
$e_2$) such that  $col_1\le \col(e)\le col_2$, where $\col(e)$ is the
color of $e$. The following result on colored range reporting
has been proved by Mortensen~\cite{Mortensen03}.
\begin{lemma}[\cite{Mortensen03}]\label{lemma:colrep}
Suppose that $\col(e)\le \log^{f}n$ for all $e\in \cL$ and some constant $f\le 1/4$. We can answer color range reporting queries on 
$\cL$ in $O(\log \log m+k)$ time using an $O(m)$-space data structure, where $m$ is the number of elements in $\cL$. Insertion of a new element into $\cL$ is  supported in $O(\log\log m)$ time. 
\end{lemma}
Note that the bound $f\le 1/4$ follows from the description in [12]: the data structure in~\cite{Mortensen03} uses Q-heaps~\cite{DBLP:journals/jcss/FredmanW94} to answer certain queries on the set of colors in constant time.

\paragraph{Colored Predecessor Problem.}
The colored predecessor query on a list $\cL$ consists of an element
$e\in \cL$ and a color $col$. The answer to  a query  $(e,col)$ is the
closest element $e'\in \cL$ which precedes $e$ such that
$\col(e)=col$.  
The following Lemma is also proved in~\cite{Mortensen03}; we also refer to~\cite{GiyoraK09}, where a similar problem is solved.
\begin{lemma}[\cite{Mortensen03}]\label{lemma:colpred}
Suppose that $\col(e)\le \log^{f}n$ for all $e\in \cL$ and some constant $f\le 1/4$. There exists an $O(m)$  space data structure that  answers colored predecessor queries on 
$\cL$ in $O(\log\log m)$ time and supports insertions in $O(\log\log m)$ time, where  $m$ is the number of elements in $\cL$. 
\end{lemma}

\subsection{Online Update of Suffix Trees for Small Alphabets}
\label{sec:onlinenew}
Classical suffix tree construction algorithms read the input text online and
spend an amortized constant time on each text letter, however
in the worst-case, they can spend as much as a linear time on an
individual letter. Several papers studied the question of reducing the
worst-case time spent on a letter, trying to approach the real-time
update
\cite{AmirKLL05,BreslauerI11,Kopelowitz12,FischerG13}. All of them
follow Weiner's algorithm and process the text right-to-left, as only
one new suffix has to be added when a new letter is prepended from
left, resulting in a constant amount of modifications. 
Breslauer and Italiano \cite{BreslauerI11} showed how to deamortize
Weiner's algorithm {\em in the case of constant-size alphabets} in order to obtain
$O(\log\log n)$ worst-case time on each new
letter. Kopelowitz~\cite{Kopelowitz12}  proposed a solution for
an arbitrary alphabet $A$ spending $O(\log\log n+\log\log |A|)$ worst-case {\em expected} time on each prepended
letter. Very recently, Fischer and Gawrychowski~\cite{FischerG13}
showed how to obtain a ({deterministic}) worst-case time 
$O(\log\log n+\frac{\log^2\log |A|}{\log\log\log |A|})$ for 
arbitrary alphabets. 

In this Section, we show a simple implementation of Weiner's algorithm
that achieves a worst-case $O(\log\log n)$ time per letter in the case
when the alphabet size is bounded by $\log^{1/4} n$. Our solution
uses Lemma~\ref{lemma:colpred} as well as a constant-time solution
to dynamic lowest common ancestor (lca) problem
\cite{DBLP:journals/siamcomp/ColeH05}. Thus, the solution below can
be viewed as a simpler and slightly more general version of the result of
\cite{BreslauerI11}, extending it from constant-size alphabets to
alphabets of size $\log^{1/4} n$. 

We first briefly recall the main idea of Weiner's algorithm using a
description similar to \cite{BreslauerI11}. Updating
a suffix tree when a new letter $a$ is prepended to the current text
$T$ is done through maintaining W-links defined as follows. 
For a suffix tree node labeled $u$ and a letter $a\in A$, W-link
$W_a(u)$ points to the locus
of string $au$ in the suffix tree, provided that $au$ is a substring
of $T$ (i.e. exists in the current suffix tree). 
Note that the locus of $au$ can be an explicit or an implicit node,
and $W_a(u)$ is called a {\em hard} or {\em soft} W-link respectively. 
The following properties of W-links will be useful in the sequel. 
\begin{lemma}[\cite{FischerG13}]\label{Wlink-props}
\begin{itemize}
\item[\textit{(i)}] If for some letter $a$, a node has a defined
  W-link $W_a$, then any its ancestor node has a defined W-link $W_a$ too.
\item[\textit{(ii)}] If two nodes $u$ and $v$ have defined hard
  W-links $W_a$, then $lca(u,v)$ has a defined hard W-link $W_a$
  too. 
\end{itemize}
\end{lemma}

When $a$ is prepended to a current text $T$, a new leaf labeled $aT$ must be
created and attached to either an existing node or a new node created
by splitting an existing edge. To find the attachment node, the
algorithm finds the lowest ancestor $u$ of the leaf labeled $T$ for which a
(possibly soft) W-link $W_a(u)$ is defined. Then the target node
$W_a(u)$ is the branching node. The main difficulty
of Weiner's approach is to find the lowest ancestor of a leaf with a
defined W-link $W_a(u)$. Another difficulty is to update (soft)
W-links when the attachment node results from an edge split (see
\cite{BreslauerI11}). 

In our solution, we store only hard W-links and do not store soft
W-links at all. Note that a hard W-link, once installed, does not need
to be updated for the rest of the algorithm
\cite{FischerG13}. Information about soft W-links is computed 
``on the fly'' using the following Lemma.

\begin{lemma}[\cite{FischerG13}]
\label{lemma:wlink1}
Assume that for a node $u$, $W_a(u)$ is defined and is a soft link pointing
to an implicit node located on an edge $(v,w)$. 
Then  there exists a unique highest descendant $u'$ of $u$ for which
$W_a(u')$ is a hard link, and, moreover, $W_a(u')=w$.
\end{lemma}

To find the lowest ancestor $u$ of a given node $t$ with a defined
(possibly soft) W-link
$W_a(u)$, consider the Euler tour of the current suffix tree in which
each internal node occurs two times corresponding to its first and
last visits. Then the following Lemma holds. 

\begin{lemma}
\label{lemma:wlink2}
Consider a node $t$. Assume that $W_a(t)$ is not defined and $u$ is
the lowest ancestor of $t$  for which a (possibly soft) link $W_a(u)$ is defined.  Let $v_1$ be the closest node preceding
$t$ in the Euler tour of the suffix tree such that $W_a(v_1)$ is a
hard link. Let $v_2$ be the closest node following $u$ in Euler tour
of the suffix tree such that $W_a(v_2)$ is a hard link. 
Then $u$ is the lowest node between $lca(t,v_1)$ and
$lca(t,v_2)$. Moreover, if $lca(t,v_1)$ is the lowest, then $v_1$ is
the highest descendant of $u$ with a defined hard W-link $W_a$,
otherwise $v_2$ is such a descendant. 
\end{lemma}
\begin{proof}
By Lemma~\ref{Wlink-props}(i), if $W_a(t)$ is not
defined, then $W_a$ is not defined for any descendant of $t$. Thus, no
node occurring between the first and the second occurrences of $t$ in
the Euler tour has a defined link $W_a$. Consequently, definitions of
nodes $v_1$ and $v_2$ are unambiguous.

By Lemma~\ref{lemma:wlink1}, $u$ has a unique closest descendant, say
$v$, with a defined hard link $W_a(v)$. If $v$ occurs before $t$ in
the Euler tour, then $v$ is the closest node preceding $t$ in the
Euler tour with defined $W_a(v)$. To show this, assume there is a closer such node
$v'$. Observe that $v'$ is also a descendant of $u$ and $v'$ is not a
descendant of $v$. By Lemma~\ref{Wlink-props}(ii), 
$lca(v,v')$ is a node with a defined hard link $W_a$. On the other
hand, $lca(v,v')$ is a proper ancestor of $v$ which is a
contradiction. Therefore, $v$ is node $v_1$ from the Lemma. 

Symmetrically, if  $v$ occurs after $t$ in the Euler tour, then $v$
is node $v_2$ from the Lemma. Clearly, to compute $u$, it is
sufficient to pick the lowest between $lca(t,v_1)$ and
$lca(t,v_2)$.
\end{proof}

Based on the above, we implement Weiner's algorithm by maintaining
the Euler tour of the current suffix tree in a colored list $\cL_W$. If a node $u$
has a defined hard W-link $W_a(u)$, then both occurrences of $u$ in $\cL_W$ are
colored with $a$. Note that a node can have up to $|A|$ hard
W-links and therefore have up to $2|A|$ occurrences in $\cL_W$. However,
the total number of hard W-links is limited by the number of tree
nodes, as a node has at most one {\em incoming} hard W-link. 

By Lemma~\ref{lemma:colpred}, we can answer colored predecessor and successor queries
on $\cL_W$ in $O(\log\log |\cL_W|)$ time. Therefore, nodes
$v_1$ and $v_2$ defined in Lemma~\ref{lemma:wlink2} can be found in
$O(\log\log |\cL_W|)$ time. Using lowest common ancestor queries on a
dynamic tree \cite{DBLP:journals/siamcomp/ColeH05}, $lca(t,v_1)$ and
$lca(t,v_2)$ can be computed in $O(1)$ time.  Therefore, updating the
suffix tree after prepending a new symbol is done is $O(\lg\lg |\cL_W|)=O(\lg\lg |T|)$
time. As an update can introduce two new hard W-links, we also need to
update the colored list $\cL_w$. This is easily done in $O(1)$ time. (Details are left out and can be found e.g. in \cite{KucherovNekrichStarikovskayaCPM12}.)

We conclude with the final result of this Section.
\begin{theorem}
\label{update-time}
Consider a text over an alphabet $A$, $|A|\le \log^{1/4}n$,
arriving online right-to-left. After prepending a new letter to
the current text $T$, the suffix tree of $T$ can be updated in time
$O(\log\log |T|)$ using an auxiliary data structure of size $O(|T|)$. 
\end{theorem}

\section{Fast Off-Line Solution}
\label{sec:algo}
In this section we describe the main part of our algorithm of
real-time text indexing. Based on the suffix tree construction from
the previous Section, the algorithm updates the
index by reading the text in the right-to-left order. 
However, the algorithm we describe in this Section will not be on-line, as it will
have to 
access symbols to the left of the currently processed symbol. 
Another ``flaw'' of the algorithm is that it will support pattern
matching queries only with an additional exception: we will be able to
report all occurrences of a pattern except for those with start
positions among a small number of most recently processed symbols of
$T$. 
In the next section we will show how to fix these issues and turn our algorithm into a
fully real-time indexing solution that reports all occurrences of a pattern.  

The algorithm distinguishes between three types of query patterns
depending on their length: {\em long patterns} contain at least
$(\log\log n)^2$ symbols, {\em medium-size patterns} contain between
$(\log^{(3)}n)^2$ and $(\log\log n)^2$ symbols, and {\em short
  patterns} contain less than $(\log^{(3)}n)^2$ symbols\footnote{Henceforth, $\log^{(3)}n=\log\log \log n$.}. 
For each of the three types of patterns, the algorithm will maintain a
separate data structure supporting queries in $O(|P|+k)$ time for matching
patterns of the corresponding type. 

\subsection{ Long Patterns}
\label{sec:long} 
To match long patterns, we maintain a {\em sparse  suffix tree} $\cT_L$
storing only suffixes that start at positions $q\cdot d$ for
$q\geq 1$ and $d=\log \log n/(4\log\sigma)$.   Suffixes stored
in $\cT_L$ are regarded as strings over a meta-alphabet of size
$\sigma^d=\log ^{1/4} n$.  This allows us to use the method of
Section~\ref{sec:onlinenew} to update $\cT_L$,
spending $O(\log\log n)$ time on each each meta-character encoding
$O(\log\log n)$ regular characters. (We recall that $\sigma=O(1)$.)

Using $\cT_L$ we can find occurrences of a pattern $P$ that start at positions 
$qd$ for $q\geq 1$, but not  occurrences starting at positions 
$qd+\delta$ for $1\le \delta <d$. 
To be able to find all occurrences, we maintain a list
$\cL_E$ defined similarly to list $\cL_W$ from Section~\ref{sec:onlinenew}. 

The list $\cL_E$ contains copies of all  nodes of $\cT_L$ as they occur during the Euler tour of $\cT_L$. 
$\cL_E$ contains one element for each leaf and two elements for each
internal node of $\cT_L$. 
If a node of $\cL_E$ is a leaf that corresponds to a suffix $T[i..]$,
we mark it with the meta-character $\Trev[i,d]=t_{i+1}t_{i+2}\ldots
t_{i+d}$ which is interpreted as the color of the leaf for the suffix $T[i..]$. 
Colors are ordered by lexicographic order of underlying
strings. If $S=s_1\ldots s_j$ is a string with $j<d$, then $S$ defines an interval
of colors, denoted $[minc(S),maxc(S)]$, corresponding to all character
strings of length $d$ with prefix $S$. 
Recall that there are $\log^{1/4} n$ different colors. 
On list $\cL_E$, we maintain the data structure of Lemma~\ref{lemma:colrep} 
for colored range reporting queries.

The update of $\cT_L$ and $\cL_E$ is done as follows. After reading character $t_{i}$ where $i=qd$ for $q\geq 1$, we add suffix 
$T[i..]$, viewed as a string over the meta-alphabet of cardinality $\log^{1/4} n$, to $\cT_L$ following the algorithm described in
Section~\ref{sec:onlinenew}. 
In addition, we have to update $\cL_E$, i.e. to insert
to $\cL_E$ the new leaf holding the suffix $T[i..]$ 
colored with $t_{i+1}t_{i+2}\ldots t_{i+d}$. (Note that at this point the
algorithm  ``looks ahead'' for the forthcoming $d$ letters of $T$.) If a new internal node
has been inserted into $\cT_L$, we also update the list $\cL_E$
accordingly.

Since the meta-alphabet size is only $\log^{1/4} n$, 
navigation  in $\cT_L$ from a node to a child can be
supported in $O(1)$  time. 
Observe that the children of any internal node $v\in \cT_L$ are naturally ordered
by the lexicographic order of edge labels. 
We store the children of $v$ in a data structure $\cP_v$ which
allows us to find in time $O(1)$ the child whose edge label starts
with a string (meta-character) $S=s_1\ldots s_d$. Moreover, we can also compute in time
$O(1)$ the ``smallest'' and the ``largest''
child of $v$ whose edge label starts with a string $S=s_1\ldots s_j$
with $j\le d$. 
$\cP_v$ will also support adding a new edge to $\cP_v$ in
$O(1)$
time. Data structure $\cP_v$ can be implemented using 
e.g. atomic heaps \cite{DBLP:journals/jcss/FredmanW94}; 
since all elements in $\cP_v$ are bounded by $\log^{1/4}n$, we can also implement $\cP_v$ as described in~\cite{NavarroN13}. 

We now consider a long query pattern $P=p_1\ldots p_m$ and show how the
occurrences of $P$ are computed. An
occurrence of $P$ is said to be a $\delta$-occurrence if it starts  in $T$ at
a position
$j=qd+\delta$, for some $q$.  For each $\delta$, $0\le \delta\le d-1$, we
find all $\delta$-occurrences as follows.  
First we ``spell out'' $P_\delta=p_{\delta+1}\ldots p_m$ in $\cT_L$
over the meta-alphabet, i.e. we traverse $\cT_L$ proceeding by blocks of up to  $d$ letters of
$\Sigma$.
If this process fails at some step, then $P$ has no
$\delta$-occurrences. Otherwise, we spell out $P_\delta$ completely,
and retrieve the closest explicit descendant node $v_\delta$, or a
range of descendant nodes $v_\delta^l, v_\delta^{l+1}, \ldots,
v_\delta^r$ in the case when $P_\delta$ spells to an explicit node
except for a suffix of length less than $d$. 
The whole spelling step takes time $O(|P|/d+1)$. 

Now we jump to the list $\cL_E$ and retrieve the first occurrence of
$v_\delta$ (or $v_\delta^l$) and the second occurrence of $v_\delta$
(or $v_\delta^r$) in $\cL_E$. 
A leaf $u$ of $\cT$ corresponds to  
a $\delta$-occurrence of $P$ if and only if $u$ occurs in the subtree
of $v_\delta$ (or the subtrees
of $v_\delta^l, \ldots, v_\delta^r$) and the color of $u$ belongs to 
$[minc(p_\delta\ldots p_1),maxc(p_\delta\ldots p_1)]$.
In the list $\cL_E$, these leaves occur precisely within the interval we
computed. 
Therefore, all
$\delta$-occurrences of $P$ can be retrieved in time $O(\log \log n +
k_\delta)$ by a colored range reporting query (Lemma~\ref{lemma:colrep}), where
$k_\delta$ is the number of $\delta$-occurrences.
Summing up over all $\delta$, all occurrences of a long pattern $P$ can 
be reported in time $O(d(|P|/d+\log\log n)+k)=O(|P|+d\log\log
n+k)=O(|P|+k)$, as $d=\log\log n/(4\log \sigma)$, $\sigma=O(1)$ and
$|P|\geq (\log\log n)^2$. 

\subsection{Medium-Size Patterns} 
\label{sec:medium}
Now we show how to answer matching queries for patterns $P$
where $(\log^{(3)}n)^2\leq |P|< (\log\log n)^2$. 
In a nutshell, we apply the same method as in Section~\ref{sec:long}
with the main difference that the sparse suffix tree will store only
 {\em truncated suffixes} of length $(\log\log n)^2$, i.e.  prefixes of suffixes 
 bounded by $(\log\log n)^2$ characters.  
We store truncated suffixes starting at positions spaced by $\log^{(3)} n=\log\log\log
n$ characters. 
The total number of different truncated suffixes 
is at most $\sigma^{(\log\log n)^2}$. This small number of suffixes
will allow us to search and update the data structures faster
compared to Section~\ref{sec:long}. 
We now describe the details of the construction. 

We store all truncated suffixes that  start at
positions $qd'$, for $q\geq 1$ and $d'=\log^{(3)} n$, in a tree $\cT_M$.  $\cT_M$ is organized
in the same  way as the standard suffix tree; that is, $\cT_M$ is a
compacted 
trie for substrings $T[qd'..qd'-(\log
\log n)^2+1]$, where these substrings are regarded as strings over the meta-alphabet $\Sigma^{d'}$.\footnote{For simplicity we assume that $\log^{(3)}n$ and $\log \log n$ are integers and $\log^{(3)}n $ divides $\log\log n$. If this is not the case, we can find $d'$ and $d$ that satisfy these requirements such that 
$\log \log n\le d\le 2\log\log n$ and $\log^{(3)}n\le d'\le
2\log^{(3)}n$.}  Observe that the same truncated suffix can occur
several times. Therefore, we augment each leaf $v$ with a list of {\em
  colors} $Col(v)$ corresponding to left contexts of the corresponding
truncated suffix $S$. More precisely, if $S=T[qd'..qd'-(\log \log
n)^2+1]$ for some $q\geq 1$, then $\Trev[qd',d']$ is added to
$Col(v)$. 
Note that the number
of colors is bounded by $\sigma^{\log^{(3)}n}$. Furthermore, for each
color $col$ in $Col(v)$, we store all positions $i=qd'$ of $T$ such that $S$
occurs at $i$ and 
$\Trev[i,d']=col$. 
Similar to Section~\ref{sec:long}, we maintain a colored list $\cL_M$ that stores
the Euler tour traversal of $\cT_M$. For
each internal node, $\cL_M$ contains two elements. 
For every leaf $v$ and for each value $col$ in its color list
$Col(v)$, $\cL_M$ contains a separate element colored with $col$.  
Observe that since the size of $\cL_M$ is bounded by $O(\sigma^{(\log
  \log n)^2+\log^{(3)}n})$, updates of $\cL_M$ can be supported in $O(\log
\log (\sigma^{(\log \log n)^2}))=O(\log^{(3)}n)$ time and 
colored reporting queries on $\cL_M$ can be answered in $O(\log^{(3)}n
+ k)$ time (Lemma~\ref{lemma:colrep}). 

Truncated suffixes are added to $\cT_M$ using a method similar to that
of Section~\ref{sec:long}.
After reading a symbol $t_{qd'}$ for some $q\geq 1$, we insert
$S_{\new}=T[qd'..qd'-(\log\log n)^2+1]$ colored with
$\Trev[qd',d']$ into the tree $\cT_M$.
Insertion of $S_{\new}$ is done as described in Section~\ref{sec:onlinenew}, and 
the list $\cL_M$ is updated accordingly. If $\cL_M$ already contains a leaf with
string value $S_{\new}$ and color $\Trev[qd',d']$, we add
$qd'$ to the list of its occurrences, otherwise  we insert a new
element  into $\cL_M$ and initialize its location list to
$qd'$. Altogether, the addition of a new truncated suffix $S_{\new}$ requires
$O(\log\log |\cT_M|) =O(\log^{(3)}n)$ time. 

A query for a pattern $P=p_1\ldots p_m$, such that $(\log^{(3)}n)^2 \leq m< (\log \log n)^2$, is answered in the same way as in Section~\ref{sec:long}. For each 
$\rho=0,\ldots,\log^{(3)}n-1$, we find locus nodes
$v_\rho^l,\ldots,v_\rho^r$ (possibly with $v_\rho^l=v_\rho^r$) of $P_\rho=p_{\rho+1}\ldots p_m$. 
Then, we find all elements in $\cL_M$ occurring between the first
occurrence of $v_\rho^l$ and the second occurrence of $v_\rho^r$ and
colored with a color $col$ that belongs to $[minc(p_\rho\ldots p_1),maxc(p_\rho\ldots p_1)]$. 
For every such element, we 
traverse the associated list of occurrences: if a position $i$ is in the list, then $P$ occurs at  position $(i+\rho)$. 
The total time needed to find all occurrences of a medium-size
pattern $P$ is $O(d'(|P|/d'+ \log^{(3)}n) +
k)=O(|P|+(\log^{(3)}n)^2+k)=O(|P|+k)$ since $|P|\geq (\log^{(3)}n)^2$.  

\subsection{Short Patterns}
\label{sec:short}
Finally, we describe our indexing data structure for patterns $P$
with $|P| < (\log^{(3)}n)^2$. 
We maintain the tree $\cT_S$ of truncated suffixes of length
$\Delta=(\log^{(3)}n)^2$ seen so far in the text. 
For every position $i$ of $T$, $\cT_S$ contains 
the substring $T[i..i-\Delta+1]$. 
$\cT_S$ is organized as a compacted trie. 
We support queries and updates 
on $\cT_S$ using tabulation. There are 
$O(2^{\sigma^{\Delta}})$ different 
trees, and $O(\sigma^{\Delta})$ different 
queries can be made on each tree.  
Therefore, we can afford explicitly storing all possible trees $\cT_S$ 
and tabulating possible tree updates. 
Each internal node of a tree stores pointers to its lefmost and
rightmost leaves, the leaves of a tree are organized in a list, and
each leaf stores the encoding of the corresponding string $Q$. 

The {\em update table} $\bT_u$ stores, for each tree $\cT_S$ and for
any string $Q$, $|Q|=\Delta$, a pointer 
to the tree $\cT'_S$ (possibly the same) obtained after adding $Q$ to $\cT_S$.
Table $\bT_u$ uses   
$O(2^{\sigma^{\Delta}} {\sigma}^{\Delta}) =o(n)$
space. 
The {\em output table}  $\bT_{o}$ stores, 
for every string $Q$ of length $\Delta$, the list of positions in the
current text $T$ where $Q$  occurs. 
$\bT_{o}$ has $\sigma^{\Delta}=o(n)$ entries and all lists of
occurrences take $O(n)$ space altogether. 

When scanning the text, 
we maintain the encoding of the string $Q$ of $\Delta$ most recently
read symbols of $T$. The encoding is
updated after each symbol using bit operations. 
After reading a new symbol, the current tree $\cT_S$ is updated using
table $\bT_u$ and the current position is added to the entry
$\bT_o[Q]$. Updates take $O(1)$ time. 

To answer a query $P$, $|P|<\Delta$, we find the locus $u$ of $P$ in the current
tree $\cT_S$, retrieve the leftmost and rightmost leaves and traverse
the leaves in the subtree of $u$. For each traversed  leaf $v_l$ with 
label $Q$, 
we report the occurrences stored in $\bT_o[Q]$. The query takes time
$O(|P|+k)$. 

\section{Real-Time Indexing}
\label{sec:realtime}
The indexes for long and medium-size patterns, described in
Sections~\ref{sec:long} and \ref{sec:medium} respectively, do not
provide real-time indexing solutions for several reasons. The
index for long patterns, for example, requires to look ahead for the
forthcoming $d$ symbols when processing symbols $t_i$ for $i=qd$, 
$q\geq 1$. Furthermore, for such $i$, we are unable to find occurrences of query
patterns $P$ starting at positions $t_{i-1} \ldots t_{i-d+1}$ before
processing $t_i$. A similar situation holds for medium-size patterns. 
Another issue is that in our previous development we assumed the
length $n$ of $T$ to be known, whereas this may of course not be the case
in the real-time setting. 
In this Section, we show how to fix these issues in order to turn the
indexes real-time. 
Firstly we show how the data structures of  
Sections~\ref{sec:long} and \ref{sec:medium} can be updated in a
real-time mode. Then, we describe how to search for patterns that
start among most recently processed symbols. 
We describe our solutions
to these issues for the case of long patterns, as a simple change
of parameters provides a solution for medium-size patterns too. 
Finally, we will show how we can circumvent the fact that the length
of $T$ is not known in advance. 

In the algorithm of Section~\ref{sec:long}, the text is partitioned
into blocks of length 
$d$,
and the insertion of a new suffix
$T[i..]$ is triggered only when the leftmost symbol $t_i$ of a block
is reached. The insertion takes time $O(d)$ and assumes the knowledge of the
forthcoming block $t_{i+d}\ldots t_{i+1}$. To turn this algorithm
real-time, we apply a standard deamortization technique. We distribute
the cost of the insertion of suffix $T[i-d..]$ over $d$ symbols of
the block $t_{i+d}\ldots t_{i+1}$. 
This is correct, as by the time we start reading the block
$t_{i+d}\ldots t_{i+1}$, we have read the block $t_{i}\ldots
t_{i-d+1}$ and therefore have all necessary information to insert suffix $T[i-d..]$.
In this way,  we spend $O(1)$ time per symbol to update all involved data structures. 

Now assume we are reading a block $t_{i+d}\ldots t_{i+1}$,
i.e. we are processing some symbol $t_{i+\delta}$ for $1\leq
\delta<i$. At this point, we are unable to find
occurrences of a query pattern $P$ starting at $t_{i+\delta}\ldots
t_{i+1}$ as well as within the two previous blocks, as they have not
been indexed yet. This concerns up to $(3d-1)$ most
recent symbols. 
We then introduce a separate procedure to search for occurrences that start
in $3d$ leftmost positions of the already processed text. This can be
done by simply storing $T$ in a compact form $T_c$ where every
$\log_{\sigma}n$ consecutive symbols are packed into one computer
word\footnote{In fact, it would suffice to store $3d-1$ most recently read symbols in compact form.}. 
Thus, $T_c$ uses $O(|T|/\log_{\sigma}n)$ words of space. Using $T_c$, we can test 
whether $T[j..j-|P|+1]=P$, for any pattern $P$ and any position $j$,
in $O(\ceil{|P|/\log_{\sigma}n})=o(|P|/d)+O(1)$ time. Therefore, checking $3d$
positions takes time $o(|P|)+O(d)=O(|P|)$ for a long pattern $P$. 

We now describe how we can apply our algorithm in the case when the
text length is not known beforehand. In this case, 
we assume $|T|$ to take increasing values $n_0<n_1<\ldots,$ as long as
the text $T$ keeps growing. Here, $n_0$ is some appropriate initial
value and 
$n_i=2n_{i-1}$ for $i\ge 1$. 

Suppose now that $n_i$ is the currently assumed value of $|T|$. 
After 
we reach character $t_{n_i/2}$,
during the
processing of the next $n_i/2$ symbols, we keep building the index
for $|T|=n_i$ and, in parallel, rebuild all the data structures under
assumption that  $|T|=n_{i+1}=2n_i$.  In particular, if $\log \log
(2n_i)\not= \log \log n_i$, we build a new index for long
patterns, and if $\log^{(3)}(2n_i)\not=\log^{(3)}n_i$, we
build a new index for meduim-size and short patterns. If
$\log_{\sigma}(2n_i)\not=\log_{\sigma}n_i$, we also construct a new
compact representation $T_c$ introduced earlier in this
section. Altogether, we distribute the construction cost of the
data structures for $T[n_i..1]$ under assumption $|T|=2n_i$ over the
processing of $t_{n_i/2+1}\ldots t_{n_i}$. 
Since $O(n_i)=O(n_i/2)$, processing these $n_i/2$ symbols
remains real-time.  By
the time $t_{n_i}$ has been read, all data
structures for $|T|=2n_i$ have been built, and the algorithm proceeds
with the new value $|T|=n_{i+1}$. Observe finally that the intervals
$[n_i/2+1,n_i]$ are all disjoint, therefore the overhead per letter
incurred by the procedure remains constant. In conclusion, the whole
algorithm remains real-time. 
We finish with our main result.
\begin{theorem}
\label{maintheorem}
There exists a data structure storing a text $T$ over a constant-size
alphabet that can be updated
in $O(1)$ worst-case time after prepending a new symbol to $T$. 
This data structure supports reporting all occurrences of a pattern $P$ in the
current text $T$ in  $O(|P|+k)$ time, where $k$ is the number of
occurrences. 
\end{theorem}

\section{Conclusions}
In this paper we presented the first real-time indexing data structure
that supports \greg{reporting all pattern occurrences in optimal time $O(|P|+k)$}. As in the previous works
on this topic~\cite{Kosaraju94,AmirN08,BreslauerI11}, we assume that
the input text is over an alphabet of constant size.  It may be
possible to extend our result to alphabets of poly-logarithmic size. 

\greg{\paragraph{Acknowledgements.} GK has been supported by the Marie-Curie
Intra-European fellowship for carrier development. We thank the
anonymous reviewers of ICALP'13 for helpful comments. }

\bibliographystyle{abbrv}

\end{document}